\documentclass[journal,12p]{IEEEtran}

\makeatletter
\def\endthebibliography{%
	\def\@noitemerr{\@latex@warning{Empty `thebibliography' environment}}%
	\endlist
}
\makeatother

\ifCLASSOPTIONcompsoc
\usepackage[nocompress]{cite}
\else
\usepackage{cite}
\fi
\ifCLASSINFOpdf
\usepackage[pdftex]{graphicx}
\else
\fi
\usepackage{cuted}
\usepackage{paralist}
\usepackage{booktabs} 
\usepackage{mathtools}
\usepackage{amsfonts}
\usepackage{graphicx,hyperref}
\usepackage{subfigure}
\usepackage{amssymb}
\usepackage{amsthm}
\usepackage{amsmath}
\usepackage{color}
\usepackage[ruled,vlined]{algorithm2e}
\usepackage{algpseudocode}

\newtheorem{theorem}{Theorem}
\newtheorem{lemma}{Lemma}
\newtheorem{proposition}{Proposition}

\newtheorem{definition}{Definition}

\newtheorem{remark}{Remark}

\renewcommand{\d}[1]{\ensuremath{\operatorname{d}\!{#1}}}
\newcommand{\ad}{\mathrm{ad}}
\newcommand{\adm}{\mathrm{adm}}
\newcommand{\Ad}{\mathrm{Ad}}

\newcommand{\expm}{{\mathrm{exp_m}}}
\newcommand{\dexp}{{\mathrm{dexp}}}
\newcommand{\dexpm}{{\mathrm{dexp_m}}}

\begin{document}
		%
		\title{{Error Dynamics} in Affine Group Systems}

		\author{
  Xinghan Li, Jianqi Chen, Han Zhang, Jieqiang Wei, and Junfeng Wu 
			\thanks{ 
X. Li is with the College of Control Science and Engineering, Zhejiang University, Hangzhou, P. R. China; Email: xinghanli0207@gmail.com. J. Chen is with the Department of Electronic and Computer Engineering, HKUST, Hong Kong; Email: eejianqichen@ust.hk. H. Zhang is with Department of Automation, Shanghai Jiao Tong University, 
shanghai, P. R. China.
J. Wei is with Ericsson, Torshamnsgatan 21, 16440, Sweden; Email:
jieqiang.wei@gmail.com. J. Wu is with the School of Data Science, CUHK-Shenzhen, Shenzhen, P. R. China; Email:  junfengwu@cuhk.edu.cn. 
			}
		}

		\maketitle
		
		\begin{abstract}
			Error dynamics captures the evolution of the state errors between two distinct trajectories that are governed by the same system rule but initiated or perturbed differently. It is essential to analyze the error behaviour of a state observer in a matrix Lie group for precise estimations. This paper concentrates on the examination of the error dynamics in affine group systems when exposed to external disturbances or random noises.
   To this end, we first discuss the connections between the notions of affine group systems and linear group systems. We characterize linear group systems by homeomorphism of its transition flow and linearity of its Lie algebra. Next, we explore linear group systems that are spread by a Brownian motion in tangent spaces. {We discover that the it\^o stochastic differential equation in the Lie algebra space has a linear drift component and an additional ``pinning'' term.} We apply these findings in analyzing the error dynamics, a specific linear group system, with smooth disturbances and stochastic noise. The error dynamics can be in the form of the Lie group and the Lie algebra. An explicit and accurate derivation of error dynamics is provided for the concrete example $SE_2(3)$, which plays a vital role, especially in inertial-based navigation.

		\end{abstract}
		

		%
		\IEEEpeerreviewmaketitle

		\section{Introduction}
		The study of the error dynamics for a state observer in matrix Lie groups is a fundamental problem in {estimation applications}, such as attitude estimation~\cite{Batista2012,Trumpf2012}, pose estimation~\cite{Zhang2022OnRO,Lei2022OptimizationOD,Zhang2021}, and localization and mapping~\cite{li2022closed,mahony_observer_2022}. {Because most of these applications are in the affine group systems~(see Definition~\ref{def:affine_system_def}) and have a symmetry structure, researchers can explore Kalman filtering based on the invariant error~\cite{Barrau}.} Such filtering methods involve a nonlinear change of variables to transform the error evolution into a linear differential equation in the absence of external disturbances and observation noises. Motivated by this, this paper focuses on examining the error dynamics of affine group systems under external disturbances or random noises. The error dynamics captures the evolution of state errors between two separate trajectories governed by the same system rule but with different initialization or perturbations. 
		We aim to exploit the evolution of the invariant error in affine groups and derive nonlinear differential equations to assist filtering design in such systems.
		
{The recent recognition of Lie group symmetry~\cite{Barrau,Goor2023EquivariantF} in the design of observers has been demonstrated to yield considerable improvements in performance when compared to an extended Kalman filter (EKF) that has not taken the symmetry into account. Multiplicative EKF (MEKF)~\cite{Lefferts1982} is one of the earliest work in this area. In particular, it focuses on filtering in the $SO(3)$ group. Subsequent research has explored similar for the special Euclidean group ($SE(3)$)\cite{Hua2011, Wang2019} and the special linear group ($SL(3)$)\cite{Hamel2011, Hua2020}. Moreover, invariant EKF (InEKF)~\cite{Barrau} introduce affine group systems, allowing higher-than-first-order kinematics in $SE_{N}(3)$. All methods require the acquisition of error dynamics for global convergence and covariance propagation. The error dynamics in the presence of disturbances can be estimated using a first-order approximation of the error without an explicit expression. This approximation, however, may lead to inaccurate estimates of the covariance~\cite{Martin2021} or even incorrect conclusions about the relationship between the estimate and the true states~\cite{Maryak2004}.}
		
		Besides, several studies have been conducted to examine stochastic processes in Lie groups~\cite{CASTELL199513,Marjanovic2018}. Barrau and Bonnabel~\cite{Barrau2018CDC} have proposed Stratonovich stochastic differential equations~(SDE's) on state within the Lie group. However, SDE's in Lie algebras are essential for analyzing statistics (e.g.  covariance) and its It\^o form is a requirement for the majority of numerical methods (e.g. Euler scheme). Marjanovi and Solo~\cite{Marjanovic2018} proposed geometric SDE's approach for matrix Lie groups, which is expressed in the form of the Lie algebra. Inspired by their work, we continue to derive the geometric SDE's of invariant errors in an affine group system under disturbances.  Smooth one and Brownian motions, are considered as the pertubation in the tangent space of the state. In particular, our main contributions are two-fold:
		\begin{enumerate}[$(i)$.]
			\item We analyze the linear group system in Proposition~\ref{pro:linear_system_property} in detail before delving into the error dynamics of the affine group system. We derive two equivalent characterizations of a linear group system using the homomorphism of its transition flow and the linearity of its Lie algebra counterpart, see Proposition~\ref{pro:linear_system_property} for greater details. 
			We further study the evolution of a linear group
			system diffused (in a left-invariant or right-invariant manner) by a Brownian motion in tangent spaces, which is written into a Stratonovich SDE. Theorem~\ref{thm:1} shows that the dynamics projected in the Lie algebra is governed by an it\^o SDE with a linear drift term and a ``pining term'', given in~\eqref{eqn:continous_time_sde}. 
			
			\item We use the previously mentioned discoveries to analyze the error dynamics between two different matrix group trajectories that are regulated by the same affine system. In this manner, error dynamics is a linear group system.  We obtain an ordinary differential equation (ODE) describing the development of the projected errors in the Lie algebra in the presence of differentiable disturbances, as stated in Theorem~\ref{thm: log_linear_specific_model}. Theorem~\ref{thm:SDE_of_Lie_logarithm} is derived to derive the counterpart with stochastic ones for errors in the form of an SDE.
			
		\end{enumerate}
		The remainder of the paper is organized as follows. Section~\ref{sec:preliminary} gives a brief preliminary of matrix Lie groups. Section~\ref{sec:affine group system} presents the main results for affine and linear group systems. Section~\ref{sec:sde}
		presents our main results in analyzing the error dynamics between two trajectories generated by an affine group system rule and in a special matrix group class $SE_N(3)$. Section~\ref{sec:conclusion} concludes the paper.
		\section{Matrix Lie Groups Preliminaries}\label{sec:preliminary}

		The paper discusses matrix Lie groups over the field of real numbers, denoted by $\mathbb{R}$, which are closed subgroups of the general linear group $GL(n,\mathbb{R})$ consisting of $n\times n$ invertible matrices with entries from $\mathbb{R}$. Every matrix Lie group is associated with a Lie algebra, whose underlying vector space is the tangent space of the Lie group at the identity element. The Lie algebra completely captures local structure of the Lie group. The matrix exponential $\mathop{\rm exp}$ is a mapping to a matrix Lie group from its Lie algebra. 
		In particular, the Lie algebra, denoted by $\mathfrak{g}$, of a Lie group $G$ can be characterized by matrix exponential as:
		\begin{equation*}
			\mathfrak{g}=\left\lbrace x\in \mathbb{R}^{n\times n}|\mathop{\rm exp}{(tx)}\in G \text{ for all } t\in \mathbb{R} \right\rbrace .
		\end{equation*}
		The Lie algebra $\mathfrak{g}$ can be identified as $\mathbb{R}^{d}$, where $d$ is the dimension of the Lie algebra $\mathfrak g$, via a ``hat'' operation; it is a linear mapping denoted as $(\cdot)^\wedge:\mathbb{R}^{d}\rightarrow\mathfrak{g}$. 
		We compose the ``hat'' mapping and the matrix exponential resulting the Lie exponential, denoted as $\expm:\mathbb{R}^{d}\rightarrow G$. The Lie logarithm, denoted by $\log_{\rm m}:G\rightarrow \mathbb{R}^{d} $, is defined as the inverse of $\rm exp_m$ within a neighborhood of the identity of $G$.
		
		Let  $\Psi_X:G\rightarrow G$ be an automorphism defined as $Y\mapsto XYX^{-1}$. With the automorphism group $\text{Aut}(G)$  of $G$, we define  $\Psi: G\rightarrow \text{Aut}(G)$ , $X\mapsto\Psi_X$. Taking the derivative of $\Psi_{X}$ at the origin, we obtain the adjoint action of a Lie group $G$, which represents the element of the group as a linear transformation of the group's Lie algebra. The adjoint action of ${X}\in G$ on
		${y}\in\mathfrak{g}$ is defined as $\Ad_{X}:\mathfrak{g}\rightarrow\mathfrak{g},  {y}\mapsto{X}{y}{X}^{-1}$. Taking the derivative of the adjoint action group at the origin, an element of the resulting Lie algebra, denoted as $\ad_{{x}}:\mathfrak{g}\to\mathfrak{g}$, is given by 
		\begin{equation*}
			\ad_{{x}}({y})=\frac{\d{}}{\d{t}}(\Ad_{\mathop{\rm exp}{(t{x})}}({y}))|_{t=0}\triangleq[{x},{y}],
		\end{equation*} which is a linear mapping from $\mathfrak{g}$ to itself. It defines an adjoint action of ${\mathfrak {g}}$ on itself. The matrix representation  of $\ad_x$ is a mapping
			from $\mathbb{R}^d$ to $\mathbb{R}^d \times \mathbb{R}^d$ such that
			$(\ad_x y)^\vee=\adm_{x^\vee} (y^\vee),$
			where $(\cdot)^\vee$ is the inverse of the ``hat'' operation.
	
		An operator $\dexp_x\triangleq\sum_{i=0}^{\infty}\frac{1}{(i+1)!}(\ad_{x})^{i}$  is commonly defined in Lie group theory to study the relation between the derivative of a matrix Lie group and the Lie algebra.
		Specifically, it is known as the left Jacobian for $\exp(x)$ when $x$ is an element of the Lie algebra of the special Euclidean group $SE(3)$~\cite{barfoot_2017}. In addition to it, the right Jacobian is defined similarly but with the sign of $x$ turned negative in ``$\ad$''. The names ``left/right Jacobian'' are overloaded in the paper, as they are used to refer to related concepts in the context of matrix Lie groups.
		The matrix representation, termed as $\dexpm$, of $\dexp$ can be defined in a similar fashion.

		\section{Linear  Group Systems and Its Stochastic Analysis}\label{sec:affine group system}	
		{In this section, we will explore SDE's that evolve in linear group system before transitioning to an exploration of error dynamics within affine group systems.}
		
		\subsection{Deterministic Dynamical System: Linear Group Systems and Affine Group Systems}
		Consider a continuous-time system evolving in $G$:
		\begin{equation}\label{eqn:continous_time_dynamics}
			\dot{X_t}\triangleq\frac{\d{}}{\d{t}}  X_{t}=f_{u_t}({X}_{t}),~~t\geq 0,
		\end{equation}
		where $u_t\in\mathcal{U}\subset\mathbb{R}^{d}$ is the input and $X_0$ is the initial state at $t=0$.  From~\eqref{eqn:continous_time_dynamics}, we know that the range of $f_{u_t}$ is the tangent space $T_{X_t}G$ at $X_t$. The trajectory $X_t$ varies when the system state is initialized differently. We denote the state transition flow from $X_0$ to $X_t$ at time $t$ as $\phi_t(X_0)$.

		Next, let us introduce the definition and properties of a linear group system.
		\begin{definition}[Linear Group System]\label{def:linear dynamics}
			The dynamical system~\eqref{eqn:continous_time_dynamics} is said to be a linear group system if $f_{u_t}$ satisfies the following property for any $u_t\in\mathcal U$ and $X,Y\in G$:
			\begin{equation}\label{eqn:linear group system}
				f_{u_t}(XY)=f_{u_t}(X)Y+Xf_{u_t}(Y) .
			\end{equation}
		\end{definition}
		
		\begin{proposition}[Properties of the linear group system]\label{pro:linear_system_property}
			Let $X_t\in G$ be the state of system~\eqref{eqn:continous_time_dynamics} and $x_t$ denote the Lie logarithm of $X_t$, i.e., $x_t\triangleq\log_{\rm m}( X_t)$. The following statements are equivalent:
			\begin{enumerate}[(i).]
				\item The system~\eqref{eqn:continous_time_dynamics} is a linear group system.
				\item It holds that $\phi_t(XY)=\phi_t(X)\phi_t(Y)$ for any $X,Y\in G$.
				\item The trajectory $x_t$ satisfies a linear ordinary differential equation
				\begin{equation}\label{eqn:dynamics_xi}
					\dot{x}_t=A_tx_t,
				\end{equation} 
				where $A_t$ is the differential of the mapping $x\mapsto(\exp_m{(x)}^{-1}f_{u_t}(\exp_m{(x)}))^\vee$ at $x=0$. 
			\end{enumerate}
		\end{proposition}
			{The proof can be found in Appendix~\ref{apx:linear_system_property}, where Lemma~1 of ~\cite{Barrau} explains how $(i)$ results in $(ii)$ and $(iii)$. Unlike the results in~\cite{Barrau}, we also investigate how a certain linear system is analogous to the linear group system.}
		
		\begin{lemma}\label{lemma:A_t property}
			In~\eqref{eqn:dynamics_xi}, $A_t$ satisfies that $$[(A_t x)^\wedge,y^\wedge]^\vee+[x^\wedge,(A_t y)^\wedge]^\vee=A_t[x^\wedge,y^\wedge]^\vee$$ for any $x,y\in\mathbb{R}^d$.		
		\end{lemma}
		\begin{proof}
			Let $X\triangleq\expm{(x)},Y\triangleq\expm{(y)}$. Then,
				\begin{align}\label{eqn:proof_A_t}
				f_{u_t}(\Ad_XY)&=f_{u_t}(\exp(\Ad_{X}y^\wedge))\notag\\
				&=(A_t(\Ad_{X}y^\wedge)^\vee)^\wedge+o(\|\Ad_Xy^\wedge\|)\notag\\
						&=(A_ty+A_t \adm_x y)^\wedge+o(\|x\|)+o(\| y\|).
			\end{align}
			By~\eqref{eqn:linear group system}, we also have
			$
			f_{u_t}(\Ad_XY)=f_{u_t}(X)YX^{-1}+Xf_{u_t}(Y)X^{-1}+XYf_{u_t}(X^{-1}).
			$ On the other hand, $f_{u_t}(\exp_{\rm m}{(x)})=(A_tx)^\wedge+o(\|x\|)$ and $\exp_{\rm m}{(x)}=I+x^\wedge+o(\|x\|)$ hold for any $x\in\mathbb{R}^d$, it writes 
			\begin{align*}
				&f_{u_t}(X)YX^{-1}+Xf_{u_t}(Y)X^{-1}+XYf_{u_t}(X^{-1})\\
			=&((A_tx)^\wedge\hspace{-1mm}+\hspace{-1mm}o(\|x\|))(I\hspace{-1mm}+\hspace{-1mm}y^\wedge\hspace{-1mm}+\hspace{-1mm}o(\|y\|))(I\hspace{-1mm}-\hspace{-1mm}x^\wedge\hspace{-1mm}+\hspace{-1mm}o(\|x\|))\hspace{-1mm}+\hspace{-1mm}\cdots\\
				=&(A_ty+\adm_{A_tx} y+\adm_{x} A_ty)^\wedge+o(\| x\|)+o(\| y\|),
			\end{align*}
			which together with~\eqref{eqn:proof_A_t} implies that  $A_t \adm_x y=\adm_{A_tx} y+\adm_{x}( A_ty)$.
		\end{proof}
	
   {
   Proposition~\ref{pro:linear_system_property} shows that the space in which the state of the linear group system evolves is homeomorphic to the space in which its Lie algebra evolves in the vector space. This allows us to focus on the linear system in the vector space instead of the linear group system. } We will now look at affine group systems and their relationships to linear group system.	
		\begin{definition}[Affine Group System \cite{Barrau}]\label{def:affine_system_def}
			The dynamical system~\eqref{eqn:continous_time_dynamics} is said to be an affine group system if $f_{u_t}$ satisfies the following property for any $u_t\in\mathcal U$, and $X,Y\in G$:
			\begin{equation}\label{eqn:affine_system_def}
				f_{u_t}(XY)=f_{u_t}(X)Y+Xf_{u_t}(Y)-Xf_{u_t}(I)Y, 
			\end{equation}
		\end{definition}
		\begin{lemma}[Connection of the two systems, Lemma~1 of \cite{Barrau}]\label{pro:linear_and_affine}
			The dynamics $\dot{X_t}=f_{u_t}(X_t)-X_tf_{u_t}(I)$ or $\dot{X_t}=f_{u_t}(X_t)-f_{u_t}(I)X_t$ is a linear group system if $f_{u_t}$ satisfies~\eqref{eqn:affine_system_def} for all $u_t\in \mathcal U$.
		\end{lemma}
		{Lemma~\ref{pro:linear_and_affine} shows us that the affine group system can be converted to the linear group system by eliminating the vector field evaluated at the identity through multiplication with the state.} 
		
		\subsection{Linear Group Systems with White Gaussian Noise}\label{sec:affine_group_system_with_nondifferentiable_noise}
		In this part, we investigate the evolution of a linear group system when it is diffused by white Gaussian noise in tangent spaces. Let $W_t$ denote an $d$-dimensional Brownian motion with derivative given by $\d{W}_t\sim\mathcal{N}(0,{I}{\d{t}})$. Here, ${W_{t,i}}$ represents the $i$-th component of $W_t$. We consider the following two stochastic processes, called left-invariant diffusion (LID) and right-invariant diffusion (RID), in $G$:
		\begin{equation}\label{eqn:dynamic_estimate_state_sde}
			\begin{split}
				\d{{X}}_t &= f_{u_t}({{X}}_t)\d{t}+\sum_{i=1}^{d}X_t s_i^\wedge\circ\d{W_{t,i}} \quad\text{(LID)},	\\
				\d{{X}}_t &= f_{u_t}({{X}}_t)\d{t}+\sum_{i=1}^{d}s_i^\wedge X_t \circ\d{W_{t,i}} \quad\text{(RID)}.
			\end{split}
		\end{equation}
		The two forms in~\eqref{eqn:dynamic_estimate_state_sde} differ in noise because a Brownian motion can either drive a left-invariant or a right-invariant vector field. The symbol $\circ$ denotes the Stratonovich calculus~\cite{ksendal1987StochasticDE}. The matrices $s_1^\wedge, \ldots,s_d^\wedge$ form an orthogonal basis of $\mathfrak{g}$, and each vector $s_i$ reflects the strength of $\d{W}_{t,i}$ applied to the vector field from respective subspaces. Therefore, the covariance of the noise can be written as $\Sigma^d_{i=1}s_i s_i^\top$.

		In stochastic integral, we may obtain different solutions depending on where we evaluate the integrated function. The It\^o and Stratonovich choices are the most natural ones, with Stratonovich's choice preserving the standard chain rule and being more intuitive for defining a stochastic process on a smooth manifold. {However, the It\^o SDE is more suitable for numerical solutions of the Euler scheme\cite{Piggott2016GeometricES} since the evaluation of the integrated function is not correlated with the increment in the It\^o integral.}
		\begin{proposition}[It\^o SDE of~\eqref{eqn:dynamic_estimate_state_sde}]\label{pro:sde_estimate_state}
			Consider~\eqref{eqn:dynamic_estimate_state_sde} in the Stratonovich sense. The It\^o forms of LID and RID are given by		\begin{equation}\label{eqn:sde_ito_left}
				\begin{split}
					\d{X}_t &= \left( f_{u_t}({{X}}_t)+\frac{1}{2}{{X}}_t\sum_{i=1}^{d}(s_i^\wedge)^2\right) \d{t}+{{X}}_t\sum_{i=1}^{d}s_i^\wedge \d{W_{t,i}},
				\end{split}
			\end{equation} 
			\begin{equation}\label{eqn:sde_ito_right}
				\begin{split}
					\d{X}_t &= \left( f_{u_t}({{X}}_t)+\frac{1}{2}\sum_{i=1}^{d}(s_i^\wedge)^2{{X}}_t\right) \d{t}
					+\sum_{i=1}^{d}s_i^\wedge {{X}}_t\d{W_{t,i}}.
				\end{split}
			\end{equation} 
		\end{proposition} 
		\begin{proof}
			We use Theorem 1.2 in \cite[CH.3]{ikeda2014stochastic} to obtain the It\^o form from Stratonovich form. For continuous semimartingales\footnote{For the sake of coherence, all discussion is this paper is limited to the case where the referred stochastic process is semimartingales.} ${A},{B},$ and ${C},$
			\begin{align}
				{B}\circ\d{{A}}&={B}\d{A}+\frac{1}{2}\d{{B}}\d{{A}},\label{eqn:proposition2_2}\\
				\left( {A}\circ\d{{B}}\right)\d{{C}}&={A}\left( \d{B}\d{C}\right).  \label{eqn:proposition2_3}
			\end{align}		
			Now consider the LID in~\eqref{eqn:dynamic_estimate_state_sde}. Substituting \eqref{eqn:proposition2_2} into the second term on the right-hand side in~\eqref{eqn:dynamic_estimate_state_sde}, we have
			\begin{equation*}
				{X}_t\sum_{i=1}^{d}s_i^\wedge\circ\d{W_{t,i}}={X}_t\sum_{i=1}^{d}s_i^\wedge\d{W_{t,i}}+\frac{1}{2}\sum_{i=1}^{d}\d{{X}_ts_i^\wedge }\d{W_{t,i}}.
			\end{equation*}
			By converting the second term on the right-hand side using~\eqref{eqn:proposition2_3}, we have 
			\begin{equation*}
				\begin{split}
					&\d{\left({{X}}_ts_i^\wedge \right) }\d{W_{t,i}}\\
					=&f_{u_t}(X_t)\d{t}s_i^\wedge\d{W_{t,i}}+{X}_t\left(\sum_{j=1}^{d}s_j^\wedge\d{W_{t,j}}s_i^\wedge\d{W_{t,i}}\right)\\
					=&{X}_t(s_i^\wedge)^2\d{t},
				\end{split}
			\end{equation*}
			where the second equality follows from  $\d{W_{t,i}}\d{W_{t,j}}=\delta_{ij}\d{t}$ and $\d{W_{t,i}}\d{t}=0$ with  $\delta_{ij}$ being a Dirac delta function. The proof for RID is similar.
		\end{proof}
		
		Based on Proposition~\ref{pro:sde_estimate_state}, we can freely convert an SDE between the It\^o and Stratonovich forms. It is worth noting that the term $\frac{1}{2}X_t \sum_{i=1}^d (s_i^{\wedge})^2{\rm d}t$ in the drift, which is often called the ``pinning drift''~\cite{Piggott2016GeometricES}, represents an infinitesimal motion in the normal space that is necessary to keep the stochastic process in $G$. The other two terms on the right-hand side of~\eqref{eqn:sde_ito_left} correspond to infinitesimal drift and diffusion movements in the tangent space.
		
		Our next result shows that the Lie logarithm, denoted as $x_t$, of the state $X_t$ in \eqref{eqn:dynamic_estimate_state_sde} is also a continuous-time Markov process driven by $W_t$. 
		When the group system~\eqref{eqn:dynamic_estimate_state_sde} contains
		a linear group drift term, the underlying stochastic process for $x_t$ is governed by an SDE with a linear drift in the vector space identifying the Lie algebra.
		
		\begin{theorem}[SDE's in the Lie algebra]\label{thm:1}
			Consider~\eqref{eqn:dynamic_estimate_state_sde} with $f_{u_t}(\cdot)$ being a linear group system. The Lie logarithm $
			x_t$ of $X_t$ obeys the following It\^o SDE if $\dexpm({-x_t})$ is invertible,  		
			
			LID:
			\begin{align}\label{eqn:continous_time_sde}
				\d{x_t} &=  \left(A_t x_t-\frac{1}{2}\sum_{k=1}^{d}{\dexp}_{\rm m}^{-1}(-x_t){C}_{t,k}\right)\d{t}+\sum_{k=1}^{d}\gamma_{t,k}\d{W_{t,k}},\notag\\
				{C}_{t,k}&=\sum_{i=1}^{\infty}\frac{(-1)^{i}}{(i+1)!}\sum^{i-1}_{r=0}\adm^{r}_{x}\adm_{\gamma_{t,k}}\adm^{i-1-r}_{x}\gamma_{t,k},\notag\\
				\gamma_{t,k} &\triangleq{\dexp}_{\rm m}^{-1}(-x_t)s_k,
			\end{align}
			
			RID:
			\begin{align}
				\d{x_t} &=  \left(A_t x_t-\frac{1}{2}\sum_{k=1}^{d}{\rm dexp}_{\rm m}^{-1}{(x_t)}{C}_{t,k}\right)\d{t}+\sum_{k=1}^{d}\gamma_{t,k}\d{W_{t,k}},\notag\\
				{C}_{t,k}&=\sum_{i=1}^{\infty}\frac{1}{(i+1)!}\sum^{i-1}_{r=0}\adm^{r}_{x}\adm_{\gamma_{t,k}}\adm^{i-1-r}_{x}\gamma_{t,k},\notag\\
				\gamma_{t,k} &\triangleq{\dexp}_{\rm m}^{-1}(x_t)s_k,
			\end{align}
			where  $A_t$ is defined in~\eqref{eqn:dynamics_xi}.
		\end{theorem}
			The full proof is given in Appendix~\ref{apx:SDE_of_Lie_logarithm}. The rationale is to use It\^o lemma to calculate each coefficient of the SDE.
   
		 Note that in Theorem  the Jacobian $\dexpm({-x})$ can be not invertible for some $x\in\mathbb R^d$. For example, when we consider the group $SO(3)$, there exist singularities for $\dexpm({-x})$ at
		$x$ with 
		$\|x\|=2k\pi$ for some nonzero integer $k$.
For a group trajectory $X_t$, we have that
			$$
			\dot{X}_t=X_t\dexp_{-x_t}(\dot{x}_t^\wedge).
			$$
The condition that the Jacobian $\dexpm({-x_t})$ is invertible validates the well-definiteness of
 ${\rm d}x_t$ in the Lie algebra.


		\section{Differential Equations of Error Dynamics in the Affine Group System}\label{sec:sde}
		
		In this section, we will examine the error dynamics in the matrix Lie group $G$, which plays a crucial role in observer design. To this end, we will investigate the evolution of the \textit{transformed errors}, defined as the Lie logarithm of the invariant errors that represent the group difference between two distinct trajectories governed by the same system rule, but initiated and perturbed differently.
		Our analysis takes into account external disturbances, which may be smooth disturbances or random noises.
		
		Consider two distinct trajectories $X_t\in G$ and $\hat{X}_t\in G$. We define two kinds of errors between the two trajectories as follows.
		\begin{definition}[Invariant Errors]~\label{def:invariant_error}
			The left- and right-invariant errors $\eta_t\in G$ of the dynamics $ X_t$ and 
			$\hat { X}_t$ are respectively defined  as follows:
			\begin{align}\label{eqn:invariant_error}
				\begin{split}
					{\eta}^{\text{L}}_t &\triangleq {\hat{X}}_t^{-1}{X}_t \quad \text{(left-invariant error~(LIE))},\\
					{\eta}^{\text{R}}_t &\triangleq {X}_t\hat{X}_t^{-1} \quad \text{(right-invariant error~(RIE))}.
				\end{split}
			\end{align}
		\end{definition}
		
These errors can also be seen as group analogs of errors in a vector space.   However, such the error representation in the matrix Lie group is not straightforward to compute. The Lie logarithm of the invariant errors in the vector space has been studied in observer design and  uncertainty propagation analysis in the Lie group. For instance, it is used in the InEKF proposed by Barrau and Bonnabel~\cite{Barrau}, and in the analysis of the concentrated Gaussian distribution evolution presented in Wang and Chirikjian's work~\cite{Wang2006ErrorPO}.

		\subsection{ODE of the Invariant Error Dynamics under Smooth Disturbance}
		We first consider a smooth disturbance $w_t\in\mathbb{R}^d$, {which is assumed to be differentiable}. We use  ${{X}}_tw_t^\wedge$ or  $w_t^\wedge{{X}}_t$ to represent the disturbance applied in the tangent space of $X_t$:
		\begin{align}
			\dot{{X}}_{t}&=f_{u_t}({{X}}_{t})+{{X}}_tw_t^\wedge, \label{eqn:continuous_noisy_model_right}\\	\dot{{X}}_{t}&=f_{u_t}({{X}}_{t})+w_t^\wedge{{X}}_t 
			\quad \text{with}\quad {X}_{t_0} ={X}_0 \label{eqn:continuous_noisy_model_left}.
		\end{align}
		
		%
		We only consider the system outlined in~\eqref{eqn:continuous_noisy_model_right}, as the same approach can be applied to the other. Consider another disturbance-free trajectory, denoted as $\hat{X}_t$, governed by the same vector field, $f_{u_t}$, but probably initialized differently:
		\begin{equation}\label{eqn:dynamics_noisy_X}
			\dot{{\hat{X}}}_t= f_{u_t}(\hat{X}_t).
		\end{equation}
		Using \eqref{eqn:continuous_noisy_model_right} and \eqref{eqn:dynamics_noisy_X}, the dynamics of $\eta^{\text{L}}$ and $\eta^{\text{R}}$  can be derived as follows:
			\begin{equation}\label{eqn:dynamic_invariant_error}
				\begin{split}
					\dot{{\eta}}^{\text{L}}_t&=f_{u_t}(\eta^{\text{L}}_t)-f_{u_t}({I})\eta^{\text{L}}_t+\eta^{\text{L}}_tw_t^\wedge \quad\text{(LIE)},\\
					\dot{{\eta}}^{\text{R}}_t&=f_{u_t}(\eta^{\text{R}}_t)-\eta^{\text{R}}_tf_{u_t}({I})+{\eta^{\text{R}}_t} \Ad_{\hat{{X}}}w_t^\wedge\quad\text{(RIE)}.
				\end{split}
			\end{equation}
			
		{Thanks to Proposition~\ref{pro:linear_and_affine}, it is easy to see that~\eqref{eqn:dynamic_invariant_error} is a linear group system when~$w^\wedge_t\equiv0$. Therefore, \eqref{eqn:dynamic_invariant_error} can be regarded as a linear group system with an additional disturbance term.}
\begin{remark}
The invariant error system~\eqref{eqn:dynamic_invariant_error} exhibits a remarkable property in that its behavior is decoupled from $X_t$ of the underlying system, instead just being determined by $u_t$ and $\hat{X}_t$. This characteristic confers an advantage in maintaining symmetry during the observer design process, which has been investigated in the works of Barrau and Bonnabel \cite{Barrau, Bonnabel2006SymmetryPreservingO}.
\end{remark}

		Let $\xi_t^L, \xi_t^R\in\mathbb{R}^{d}$ denote Lie logarithm of the left-invariant and right-invariant error, i.e., the transformed error, which writes $\xi_t^L\triangleq\log_{\rm m}(\eta_t^L)$ and $\xi_t^R\triangleq\log_{\rm m}(\eta_t^R)$. In the sequel, we will drop the superscript from the two notations when there is no need to specify which error we are working with.
		Next, we describe our findings regarding the dynamics of the $\xi_t$ that evolves in the Lie algebra.
		\begin{theorem}[{Error dynamics with smooth disturbance}]\label{thm: log_linear_specific_model}
			We consider~\eqref{eqn:dynamic_invariant_error} and assume that $\dexpm(-\xi_t)$ is invertible during $t\in[t_0,T)$, where $T\in(t_0,\infty]$. Then, the dynamics of $\xi_t$ is given by
			\begin{align}
				\dot{\xi}^{\text{L}}_t&=\dexpm^{-1}(-\xi_t^{\text{L}})w_t+A_t\xi^{\text{L}}_t\quad \text{(LIE)},\\
				\dot{\xi}^{\text{R}}_t&=\dexpm^{-1}(-\xi_t^{\text{R}})\Ad_{\hat{{X}}_t}w_t+A_t\xi^{\text{R}}_t\quad \text{(RIE)},
			\end{align}
			where  $A_t$ is defined in~\eqref{eqn:dynamics_xi}.
		\end{theorem}
		\begin{proof}
			We only prove for the left-invariant error case. The superscript $L$ is omitted from $\eta_t^L$ and $\xi_t^L$. By
			Theorem 5 of \cite{Hunacek2008}, we have that 
			$
			(\eta^{-1}_t\dot{\eta_t})^\vee= \dexpm(-\xi_t)\dot{\xi_t}.
			$
			Since $\dexpm(-\xi_t)$ is bijective by assumption,       
			\begin{equation*}
				\begin{split}
					\dot{\xi_t} &= \dexp^{-1}_{\rm m}(-\xi_t)	\left( {\eta_t}^{-1}f_{u_t}(\eta_t)-{\eta_t}^{-1}f_{u_t}({I})\eta_t+w_t^\wedge\right)^\vee\\
					&\triangleq g_{u_t}(\xi_t)+\dexp^{-1}_{\rm m}(-\xi_t)w_t.
				\end{split}
			\end{equation*}
			Next we will show that $g_{u_t}$ is a linear mapping. 
			Consider a trajectory $\eta_t'\in G$ with dynamics $\dot{\eta_t}'=f_{u_t}(\eta_t')-f_{u_t}(I)\eta_t'$ and its logarithm $\xi'_t$. 
			Based on the derivative of the exponential, the dynamics of ${\xi}_t'$ writes
			\begin{equation*}
				\dot{\xi_t'}=\dexpm^{-1}{(-\xi'_{t})}(\eta'^{-1}_t\dot{\eta'}_t)^\vee=g_{u_t}(\xi'_t).
			\end{equation*}  
			By~Proposition~\ref{pro:linear_and_affine},  $\xi'_t$ follows dynamics
			$\dot{\xi'_t}={A}_t\xi'_t$. Therefore,
			we conclude that $g_{u_t}(\xi_t')={A}_t\xi_t'$,
			which completes the proof.
		\end{proof}
		
		\subsection{SDE's of the Invariant Error Dynamics under Brownian Motion}
		
		Consider $X_t$ and  $\hat{X}_t$ generated by~\eqref{eqn:sde_ito_left}\footnote{We only focus on the LID case in~\eqref{eqn:sde_ito_left} as the results of the other case
			can be derived similarly.} and~\eqref{eqn:continous_time_dynamics}, respectively, which write
\begin{equation}\label{eqn:estimator}
			\begin{split}
				\d{X}_t &= \left( f_{u_t}({{X}}_t)+\frac{1}{2}{{X}}_t\sum_{k=1}^{d}(s_k^\wedge)^2\right) \d{t}+{{X}}_t\sum_{k=1}^{d}s_k^\wedge \d{W_{t,k}},\\
				\d{\hat{X}}_t&=f_{u_t}(\hat{X}_t) \d{t}.
			\end{split}
		\end{equation} 
		\begin{lemma}\label{lemma:invariant error dynamics}
			The invariant errors $\eta_t$ of~\eqref{eqn:estimator} satisfies the following SDE:
			\begin{equation}\label{eqn:sde_invariant_error}
				\begin{split}
					\d{\eta_t}^{\text{L}}&=  g_{u_t}(\eta^{\text{L}}_t)\d{t}+\eta^{\text{L}}_t\sum_{k=1}^{d}\left(\frac{1}{2}(s_k^\wedge)^2 \d{t}+s_k^\wedge d{W_{t,k}}\right),\\
					\d{\eta_t}^{\text{R}}&=  g_{u_t}(\eta^{\text{R}}_t)\d{t}+\eta^{\text{R}}_t\Ad_{\hat{X}_t}\sum_{k=1}^{d}\left(\frac{1}{2}(s_k^\wedge)^2\d{t}+s_k^\wedge d{W_{t,k}}\right).
				\end{split}
			\end{equation} 
		\end{lemma}
		\begin{proof}
			We derive the SDE of the left-invariant error, and for the right-invariant one it follows similarly:
			\begin{equation*}
				\begin{split}
					\d{\eta}^{\text{L}}_t&=\d{(\hat{X}^{-1}_t{X}_t )}=\d{ (\hat{X}^{-1}_t )}{X}_t+\hat{X}^{-1}_t\d{ ({X}_t) }\\
					&=- \hat{X}^{-1}_t\d{ ( \hat{X}_t) } \hat{X}^{-1}_t{X}_t+\hat{X}^{-1}_t\d(X_t) \\
					&=  g_{u_t}(\eta^{\text{L}}_t)\d{t}+\eta^{\text{L}}_t\sum_{k=1}^{d}\left(\frac{1}{2}(s_k^\wedge)^2 \d{t}+s_k^\wedge d{W_{t,k}}\right),
				\end{split}
			\end{equation*}
			where the last equality holds due to~\eqref{eqn:sde_ito_left}, which completes the proof.
		\end{proof}
		From Lemma~\ref{lemma:invariant error dynamics}, the diffusions of ${\eta}^{\text{L}}_t$ and
		${\eta}^{\text{R}}_t$ can be seemed as the left-invariant stochastic processes in linear group system. By
	 Theorem~\ref{thm:1}, we can show that the dynamics of 
		the transformed error $\xi_t$ is characterized by an SDE too, which is detailed in the following theorem.
		\begin{theorem}[{Error dynamics with nondifferentiable noise}]\label{thm:SDE_of_Lie_logarithm}
		 The Lie logarithm of the invariant errors $\xi_t$ for \eqref{eqn:estimator} obeys an SDE if
		 $\dexpm(-\xi_t)$ is invertible, 
\begin{equation}\label{eqn:therem_3}
				\d{\xi_t} =  \left(A_t \xi_t-\frac{1}{2}\sum_{k=1}^{d}\dexpm^{-1}{(-\xi_{t})}{C}_{t,k}\right)\d{t}+\sum_{k=1}^{d}\gamma_{t,k}\d{W_{t,k}}
			\end{equation}
			where 
			\begin{align*}		{C}_{t,k}&=\sum_{i=1}^{\infty}\frac{(-1)^{i}}{(i+1)!}\sum^{i-1}_{r=0}\adm^{r}_{\xi}\adm_{\gamma_{t,k}}\adm^{i-1-r}_{\xi}\gamma_{t,k}\end{align*} 
			$\gamma_{t,k}\triangleq \dexpm^{-1}{(-\xi^{\text{L}}_{t})}s_k$ for LIE, $\gamma_{t,k} \triangleq \dexpm^{-1}{(-\xi^{\text{R}}_{t})}\Ad_{\hat{X}_t}s_k$ for RIE.
		\end{theorem}
		
		\begin{remark}
			In ~\cite{Barrau},  Barrau and Bonnabel proposed an approximate error dynamics for studying the stability of the estimator. The approximation was obtained by neglecting higher-order terms of $\xi_t$ in $\dexpm{(-\xi_t)}$ of LIE of \eqref{eqn:therem_3}, which can be written into
				\begin{equation}~\label{eqn:linear approximation}
					\mathrm{d}\xi_t = A_t\xi_t\d{t}+\sum_{k=1}^{d}s^\wedge_k\mathrm{d}W_{t,k}.
				\end{equation}   
		\end{remark}
		
		\subsection{Examples on $SE_2(3)$ }
  {
		The matrix group $SE_2(3)$ is frequently employed in industrial settings, for instance in navigation with inertial measurement units~(IMU) as described in \cite{Barrau2020}.   
		A state $X_t\in SE_2(3)$ consists of a rotation matrix ${R}\in SO(3)$ and $2$ vectors $t_1, t_2\in\mathbb R^3$, given by the form:
		$\begin{bmatrix}
			{R}&\vline& {t}_1& {t}_2\\
			\hline
			{0} &\vline	&
			\begin{matrix}
				I
			\end{matrix}
		\end{bmatrix}.
		$ 		
		The matrix representation of the corresponding Lie algebra $\mathfrak{se}_2(3)$ is given by
		$
\begin{bmatrix}
\omega^\wedge &\vline& {v}_1& {v}_2\\
			\hline
			{0} &\vline	&
			\begin{matrix}
				0
			\end{matrix}
		\end{bmatrix}.$ The kinematic system of the IMU with $X_t\in SE_2(3)$ is modeled as
\begin{align*}
    \d{X}_t = (X_tv^\wedge_b+v^\wedge_gX_t+f_0(X_t))\d{t}+{{X}}_t\sum_{k=1}^{d}s_k^\wedge \d{W_{t,k}}
\end{align*}
  as \cite{li2022closed} has done for IMU and $f_0()$ denote a linear group system which represents $\dot{t}_1=t_2$. In addition, we also consider an estimator without random noise.	Let $\xi$ denote the transformed error between the estimator and the state. We can use \eqref{eqn:therem_3} to provide the error dynamics and provide $C_{t,k}$ and $\dexpm{(-\xi_{t})}$ in the following.}
\setcounter{equation}{24}
\begin{figure*}[t]
\begin{equation}\label{eqn:ck}
\begin{split}
    C_{t,k}&=\sum^4_{j=1}\beta'_{t,j}\theta_t'\adm^{j}_{\xi_t}\gamma_{t,k}+\beta_j\sum^{j-1}_{r=0}\adm^{r}_{\xi_t}\adm_{\gamma_{t,k}}\adm^{j-1-r}_{\xi_t},\quad\theta_t'\triangleq\begin{bmatrix}
    \omega_t^\top&0&0
\end{bmatrix}\gamma_{t,k}\\
    \beta'_{t,1}&\triangleq\frac{16-16\cos{\theta_t}+2\theta_t\cos{\theta_t}-5\theta_t\sin{\theta_t}}{4\theta_t},\quad\beta'_{t,2}\triangleq\frac{15\sin{\theta_t}-8\theta_t-7\theta_t\cos{\theta_t}-\theta_t^2\sin{\theta_t}}{2\theta_t^4}\\
   \beta'_{t,3}&\triangleq\frac{8-5\theta_t\sin\theta_t+\theta_t^2\cos\theta_t-8\cos\theta_t}{2\theta_t^5},\quad\beta'_{t,4}\triangleq\frac{15\sin\theta_t-\theta_t^2\sin\theta_t-7\theta_t\cos\theta_t-8\theta_t}{2\theta_t^6}
\end{split}
\end{equation}
\hrulefill
\end{figure*}
\setcounter{equation}{25} 
  \begin{proposition}[Closed-expression for $SE_2(3)$]\label{proposition:Lie_log_evolution_SE3}
			In~\eqref{eqn:therem_3}, ${C}_{t,k}$ and the right Jacobian $ \dexpm{(-\xi_{t})}$ are given by
			\begin{equation}\label{eqn:se_n3_equation_1}
				\begin{split}
			&\dexpm(-\xi_t)=I+\sum_{j=1}^{4}\beta_{t,j}\adm^{j}_{\xi_t}\\
					&\beta_{t,1}=\frac{4 \cos\theta_t-4+\theta_t \sin \theta_t}{2 \theta_t^2}\quad\beta_{t,2}=\frac{4 \theta_t-5 \sin \theta_t+\theta_t \cos \theta_t}{2 \theta_t^3}\\
					&\beta_{t,3} = \frac{\theta_t \sin \theta_t+2 \cos \theta_t-2}{2 \theta_t^4}\quad\beta_{t,4}=\frac{2 \theta _t-3 \sin \theta_t+\theta_t \cos \theta_t}{2 \theta_t^5}
				\end{split}
			\end{equation}
		where $\theta_t\triangleq\| \omega_t \|$ and $C_{t,k}$ is evaluated in~\eqref{eqn:ck}.
		\end{proposition}
		\begin{proof}
			For any $\xi^\wedge\in \mathfrak{se}_2(3)$, we have the identity $\ad^5_{\xi^\wedge}+2\theta^2\ad^3_{\xi^\wedge}+\theta^4\ad_{\xi^\wedge}\equiv 0 $, by (7.62) of~\cite{Hall:371445}. Using this identity, we express $\dexpm(-\xi_t)$ as a power sum of $\adm_{\xi_t}$ up to the fourth order term, given by ~\eqref{eqn:se_n3_equation_1}.
			The result of $C_{t,k}$ is obtained by differentiating the right hand side of \eqref{eqn:se_n3_equation_1} using the product rule. By the chain rule,  we obtain $	\sum_{k=1}^{d} \left(\partial_{\xi_{i}}(\adm_{\xi})^j\right)$ and $\partial_{\xi_{j}}\theta$ to evaluate,  where $	\sum_{k=1}^{d} \left(\partial_{\xi_{i}}(\adm_{\xi})^j\right)$ is evaluated in Lemma~\ref{lemma:support_lemma_2}. Next we briefly provide the evaluation of $\partial_{\xi_{i}}\theta$ by using the fact $\theta^2=\omega^\top\omega$. Differentiating both sides with respect to $\xi_{i}$ and re-arranging, we get that $\partial_{\xi_{i}}\theta=\frac{\omega_i}{\theta}$, from which the values of  $C_{t,k}$'s can be solved.
\end{proof}
{
\subsection{Experiments}
 We employ simulation to compare the statistics of transformed errors obtained from SDE's on state within $SE_2(3)$~\eqref{eqn:estimator}($SE_2(3)$-SDE), SDE's on the error within the vector space $\mathbb{R}^9$~\eqref{eqn:therem_3}(${R}^9$-SDE), and the linear approximation via~\eqref{eqn:linear approximation}(linear-SDE). We use the Euler–Maruyama~\cite{Piggott2016GeometricES} technique to numerically solve SDE's and 1000 samples for statistics. The parameters $v_b$ and $v_g$ are set to $[0;0;0.5;0;0;0;0;0;9.81]$ and $[0;0;0;0;0;0;0;0;-9.81]$ respectively, and the angular velocity and acceleration noise are set to $0.01$ and $0.1$, respectively, to simulate the kinematic system of the IMU.
\begin{figure}[htbp]
    \centering
    \includegraphics[width=0.5\textwidth]{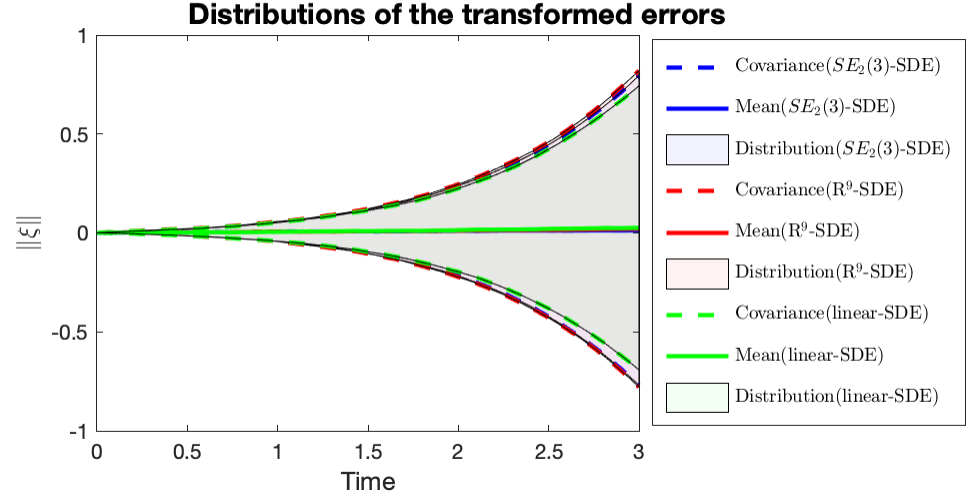}
    \caption{Evolution of errors for three methods, which are obtained by SDE's on state, the transformed error, and linear approximation. The means of all methods are close to zero, and the covariance calculated by the linear approximation is less than the other two, resulting in an inconsistent covariance during the propagation.}
    \label{fig:error_evo}
\end{figure}

Figure~\ref{fig:error_evo} show that the error statistics obtained from $SE_2(3)$-SDE and ${R}^9$-SDE~\eqref{eqn:therem_3} are consistent, while linear-SDE produces a covariance over time that is smaller than the other two and not consistent.
}

			\section{Conclusions}\label{sec:conclusion}
			
			In this paper, we analyzed the diffusion processes in the linear group system and provided an It\^o SDE description for the projected dynamics evolving in the Lie algebra. By the previously analysis, we studied the development of the invariant error of two trajectories in the presence of disturbances, and showed the error dynamics in both ODE's and SDE's forms. Our explicit and accurate derivation of error dynamics for a concrete matrix group $SE_2(3)$ may have practical applications in estimation of the IMU.
			Estimation tasks often necessitate the use of filtering in affine group systems with stochastic dynamics, where precise estimation is contingent upon a thorough comprehension of our derivations.
			

			\appendices \section{Proof of Proposition~\ref{pro:linear_system_property} }\label{apx:linear_system_property}		

				We introduce the following supporting lemma for the proof.
				
					%
					%
				
				\begin{lemma}\label{lemma:supporting_lemma_of_proposition_1}
					Consider  a linear dynamical system \begin{equation}\label{eqn:LTV_dynamical_system}
						\dot{y}_t=A_t y_t ,~t\geq 0,
					\end{equation}
					in 
					$\mathfrak g$, where $\mathfrak g$  is the Lie algebra of a connected Lie group $G$, with  $y^\wedge_t\in  \mathfrak g$ and a given initial condition $y_{t_0}= y_0$. The linear mapping $A_t$ satisfies $[(A_t y_1)^\wedge,(A_t y_2)^\wedge]^\vee=A_t([y_1^\wedge,y_2^\wedge])^\vee$ mentioned in Lemma~\ref{lemma:A_t property}. Let $\Phi_{t,t_0}$ be the state transition mapping, i.e., $\Phi_{t,t_0}: y_{t_0}\mapsto y_t$. Then $\Phi_{t,t_0}$
					has the following properties:
					\begin{enumerate}[(i).]
						\item The dynamics of $\Phi_{t,t_0}$, $t>0$, satisfies a differential equation $\dot{\Phi}_{t,t_0}=A_t\Phi_{t,t_0}$ with initial state $\Phi_{t_0,t_0}$ being the identity mapping.
						\item  For any $\Phi_{t,t_0}$, $t>0$ and $ y,y'\in\mathbb{R}^d$,   $\Phi_{t,t_0}$ is a linear map satisfying 
						\begin{equation*}
							\Phi_{t,t_0}[y^\wedge,y'^\wedge]^\vee=[(\Phi_{t,t_0}y)^\wedge,(\Phi_{t,t_0}y')^\wedge]^\vee.
						\end{equation*}
					\end{enumerate}
				\end{lemma}
				\begin{proof}
					The first property follows from Theorem 5.1 of~\cite{hespanha2018linear}. We will prove the second property.
					Given any $y,y'\in \mathfrak{g}$ and consider the following three trajectories:
					\begin{align*}
						y_{0,t}= \Phi_{t,t_0}y_{0,t_0},~
						y_{1,t}= \Phi_{t,t_0}y_{1,t_0},~\hbox{and~}
						y_{2,t}= \Phi_{t,t_0}y_{2,t_0},
					\end{align*}
					where $y^\wedge_{0,t_0}=[y^\wedge,y'^\wedge]$, 
					$y^\wedge_{1,t_0}=y^\wedge$ and $y^\wedge_{2,t_0}=y'^\wedge$, respectively. It is ready to show that the three trajectories all follow the linear dynamical system~\eqref{eqn:LTV_dynamical_system}.
					Let $z^\wedge_t=[y^\wedge_{1,t},y^\wedge_{2,t}]$. Then, for $t\geq 0$,
					\begin{align*}
						\dot{z}_t^\wedge&=[\dot{y}_{1,t}^\wedge,y_{2,t}^\wedge]+[y_{1,t}^\wedge,\dot{y}_{2,t}^\wedge]\\
						&=[(A_t y_{1,t})^\wedge,y^\wedge_{2,t}]+[y^\wedge_{1,t},(A_ty_{2,t})^\wedge]\\
						&=(A_t[y^\wedge_{1,t},y^\wedge_{2,t}]^\vee)^\wedge=(A_tz_t)^\wedge,
					\end{align*}
					where the third equality uses $[(A_t y_1)^\wedge,(A_t y_2)^\wedge]=(A_t([y_1^\wedge,y_2^\wedge])^\vee)^\wedge$. 
					Since $z_{t_0}=[y^\wedge,y'^\wedge]^\vee=y_{0,t_0}$, we have 
					$y_{0,t}=z_t$ for all $t\geq 0$, that is,
					$$
				\Phi_{t,t_0}[y^\wedge,y'^\wedge]^\vee=[(\Phi_{t,t_0}y)^\wedge,(\Phi_{t,t_0}y')^\wedge]^\vee,
					$$
					which completes the proof.		
				\end{proof}
				
				
				Now we are ready to prove Proposition~\ref{pro:linear_system_property}.
				
				$(i)\Rightarrow (iii)$. 
				Theorem 7 in~\cite{Barrau} shows that if the state $X_t$ of the dynamical system in \eqref{eqn:linear group system} satisfies a linear group system, then its Lie logarithm $x_t$ satisfies a linear system $\dot{x_t}=A_tx_t$, where $A_t$ is calculated using $f_{u_t}(\expm{(x_t)})=(A_t x_t)^\wedge +O(\|x\|^2)$. The sketch of the proof is that the solution $\phi_t(\expm{(x_0)})$ of a linear group system with initial condition $\expm{(x_0)}$ satisfies $\phi_t(\expm{(x_0)})=\expm{(\Phi_tx_0)}$ and Barrau and Bonnabel demonstrate that $\Phi_t$ is also the state transition matrix of a linear system, i.e., $\frac{\text{d}}{\d{t}}\Phi_t=A_t\Phi_t$. See \cite{Barrau} for more details. 
				
				$(iii)\Rightarrow(ii)$. Consider the dynamics of $x_t$ in $(iii)$.
				Let $X_t=\expm{\left( x_t\right) }$. For any $X,Y\in G$, we can find $x_a, x_b\in \mathbb R^d$ such that $\expm(x_a)=X$ and $\expm(x_b)=Y$. Therefore, 
				\begin{align*}
					&\phi_t(X)\phi_t(Y)=\phi_t(\expm(x_a) )\phi_t(\expm( x_b))\\
					&=\expm( \Phi_{t,0}x_a) \expm( \Phi_{t,0}x_b)= \expm\left(\Phi_{t,0}( \text{BCH}(x_a ^\wedge,x_b^\wedge)^\vee\right) \\		
					&=\phi_t\left(\exp\left( \text{BCH}( x_a ^\wedge, x_b^\wedge)\right)\right)= \phi_t(XY),
				\end{align*}
				where  the fourth equality uses  Lemma~\ref{lemma:supporting_lemma_of_proposition_1}, and $\Phi_{t,0}$ is the state transition mapping discussed in Lemma~\ref{lemma:supporting_lemma_of_proposition_1} and $\text{BCH}(x_a^\wedge,x_b^\wedge)\triangleq x_a^\wedge+x_b^\wedge+\frac{1}{2}[x_a^\wedge,x_b^\wedge]+\cdots$ is the Baker–Campbell–Hausdorff formula for $x_a^\wedge,x_b^\wedge\in\mathfrak{g}$.
				
				$(ii)\Rightarrow (i)$. 	 For any $X,Y\in G$ and $t>0$, we have 
				\begin{align*}
					f_{u_t}(\phi_t(XY))&=\frac{\d{}}{\d{t}}\phi_t(XY)=\frac{\d{}}{\d{t}}\left( \phi_t(X)\phi_t(Y)\right)\\ 
					&=f_{u_t}(\phi_t(X))\phi_t(Y)+\phi_t(X) f_{u_t}(\phi_t(Y)).
				\end{align*}
				Therefore we have $f_{u_t}(XY)=f_{u_t}(X)Y+Xf_{u_t}(Y)$ for $t=0$,
				which completes the proof.

				\section{Proof of Theorem~\ref{thm:1}}\label{apx:SDE_of_Lie_logarithm}
				We only prove for the LID case; and for the RID case the proof is similar.
				The proof is divided into two parts. The first part uses It\^o's lemma to derive an SDE with undetermined coefficients and the second part involves calculating the coefficients from Proposition~\ref{pro:sde_estimate_state}.
				
				Suppose that the evolution of $x_t\in\mathbb R^d$ follows a formal SDE:
				\begin{equation}\label{eqn:xt}
					\begin{split}
						\d{x_t}&={F}_t\d{t}+{H}_t\d{W_t},
					\end{split}
				\end{equation}
				with undetermined coefficients ${F}_t\in\mathbb{R}^{d}$ and  ${H}_t\in\mathbb{R}^{d\times d}$.
				Since $X_t=\expm{\left(x_t\right)}\triangleq\exp{\left( \sum_{i=1}^{d}x_{t,i}{E}_i \right)}$, where $E_i$ is a standard orthogonal basis of the Lie algebra, the exact expression for $\d{X}_t$ can be given by  It\^o lemma with two terms:
				\begin{equation}\label{eqn:48}			\d{{X}_t}=\sum_{i=1}^{d}\d{x_{t,i}}\partial_{x_{t,i}}X+\frac{1}{2}\sum_{j=1}^{d}\sum_{i=1}^{d}\d{x_{t,j}}\d{x_{t,i}}\partial^2_{x_{t,ji}}X,
				\end{equation} 
				where $\partial_{x_{t,i}}X\triangleq 
				\frac{\partial \expm(x)}{\partial{x_{i}}}|_{x=x_t}$
				and $\partial^2_{x_{t,ji}}X\triangleq \frac{\partial^2 \expm(x)}{\partial x_{j}\partial x_{i}}|_{x=x_t}$.
				We first go into detail about $\d{x_{t,j}}\d{x_{t,i}}$,
				\begin{align*}
					&\d{x_{t,j}}\d{x_{t,i}}\\
					=&( F_{t,j}\d{t}+\sum_{k=1}^{d}H_{t,jk}\d{W}_{t,k}) ( F_{t,i}\d{t}+\sum_{k=1}^{d}H_{t,ik}\d{W}_{t,k}) \\
					=&\sum_{k=1}^{d}H_{t,jk}H_{t,ik}\d{t}.
				\end{align*}
				Substituting this and~\eqref{eqn:xt} into~\eqref{eqn:48}, we obtain that		\begin{equation}\label{eqn:SDE_M_N}
					\d{X}_t={M}_t\d{t}+{N}_t \d{W_t}
				\end{equation}
				with
				\begin{align*}
					{M}_t&\triangleq\sum_{i=1}^{d}F_{t,i}\partial_{x_{t,i}}X+ \frac{1}{2}\sum_{i,j,k=1}^{d}H_{t,jk}H_{t,ik}\partial^2_{x_{t,ji}}X,\\
					{N}_t&\triangleq\sum_{i=1}^{d}\sum_{k=1}^{d}H_{t,ik}\partial_{x_{t,i}}X.
				\end{align*}
				By Theorem 5 of~\cite{Hunacek2008}, we find 
				$
				\partial_{x_{t,i}}X=X_t \dexpm_{-x_t}{E}_i
				$, which further yields
				\begin{equation*}
					\begin{split}
						\sum_{i=1}^{d}F_{t,i}\partial_{x_{t,i}}X&= X_t\dexp_{-x_t}\sum_{i=1}^{d}F_{t,i}{E}_i= X_t\dexp_{-x_t}({F_t}^\wedge),\\
						\sum_{i=1}^{d}\sum_{k=1}^{d}H_{t,{ik}}\partial_{x_{t,i}}X&= X_t\dexp_{-x_t}\sum_{k=1}^{d}\text{col}_k({H})^\wedge.
					\end{split}
				\end{equation*}
				We next calculate $\Delta_{t,k}\triangleq\sum_{j=1}^{d}\sum_{i=1}^{d}H_{t,jk}H_{t,ik}\partial^2_{x_{t,ji}}X$ in the next lemma,
				the proof of which is further presented in the Appendix~\ref{apx:lemma_computation}.
				\begin{lemma}\label{lemma:calculate_support_lemma}
						$$
						\Delta_{t,k}=X_t\left(  \dexp_{-x_t}(\gamma^\wedge_{t,k})\right) ^2+X_t{C}^\wedge_{t,k},
						$$
						where ${C}_{t,k}\triangleq\sum_{i=0}^{\infty}\frac{(-1)^{i+1}}{(i+2)!}\ad_{\gamma^\wedge_{t,k}}(\ad_{x_t^\wedge})^i\gamma_{t,k}$
						with $\gamma_{t,k}\triangleq\text{col}_k({H_t})$.
					\end{lemma}

				
				On the other hand, from Proposition~\ref{pro:linear_system_property}, we have
				\begin{equation*}
					f_{u_t}(X_t)=X_t\dexp_{-x_t}((A_tx_t)^\wedge).
				\end{equation*}
				Comparing~\eqref{eqn:dynamic_estimate_state_sde} and~\eqref{eqn:SDE_M_N} term by term, it implies that 
					\begin{equation*}
						\begin{split}
					{F}_t&=A_t x_t-\frac{1}{2}\dexpm^{-1}(-x_t)\sum_{k=1}^{d}{C}_{t,k},\\
					\text{col}_k({H}_t)&=\dexpm^{-1}(-x_t)s_k,
						\end{split}
					\end{equation*}
				which completes the proof.
				\hfill $\square$
				
				\section{Proof of Lemma~\ref{lemma:calculate_support_lemma}}\label{apx:lemma_computation}
						\begin{lemma}\label{lemma:support_lemma_2}
					$$			\sum_{j=1}^{d}H_{t,jk} \left(\partial_{x_{j}}(\adm_{x})^n\right)=\sum^{n-1}_{r=0}\adm^{r}_{x}\adm_{\gamma_{t,k}}\adm^{n-1-r}_{x},$$   for any $x\in\mathbb R^d$.
					
				\end{lemma}
				\begin{proof}
					The proof is by induction. Using the definition of $\ad$, for the case $n=1$ and any $y\in\mathfrak g$,
					\begin{align*}
						\sum_{j=1}^{d}H_{t,jk} \partial_{x_j}(\ad_{x^\wedge}y)
						&= \sum_{j=1}^{d}H_{t,jk}\partial_{x_j}[x,y]\\
						&= \sum_{j=1}^{d}H_{t,jk}\left( {E}_jy-y{E}_j\right)=\ad_{\gamma^\wedge_{t,k}}y.
					\end{align*}
					Assume the result holds for the case $n$, we show it holds for the case $n+1$. The $n+1$ case writes
					\begin{align*}
	 \partial_{x_j}(\ad^{n+1}_{x^\wedge}y)&= \partial_{x_j}\left( x^\wedge\ad^{n}_{x^\wedge}y-\ad^{n}_{x^\wedge}yx^\wedge\right)\\
	 &=\ad_{E_j}\ad^{n}_{x^\wedge}y+\ad_{x^\wedge} \partial_{x_j}(\ad^{n}_{x^\wedge}y)\\
	 &=\ad_{E_j}\ad^{n}_{x^\wedge}y+\ad_{x^\wedge}\sum^{n-1}_{r=0}\ad^{r}_{x^\wedge}\ad_{E_j}\ad^{n-1-r}_{x^\wedge}y\\
	 &=\sum^{n}_{r=0}\ad^{r}_{x^\wedge}\ad_{E_j}\ad^{n-r}_{x^\wedge}y,
					\end{align*}
			where the RHS matches the result for the case $n+1$. This completes the proof.
				\end{proof}
				
			\fbox{Proof of Lemma~\ref{lemma:calculate_support_lemma}}
We have that
				\begin{equation*}
					\begin{split}
						&\sum_{j=1}^{d}\sum_{i=1}^{d}H_{t,jk}H_{t,ik}\partial_{x_{t,ji}}^2X_t\\
						=&\sum_{j=1}^{d}H_{t,jk}\partial_{x_{t,j}}( X_t\dexp_{-x_t})(\gamma_{t,k}^\wedge)\\
						=& \sum_{j=1}^{d}H_{j,k}\left( (\partial_{x_j}X_t)\dexp_{-x_t}+X_t (\partial_{x_j}\dexp_{-x_t}) \right)(\gamma^\wedge_{t,k})\\
						=&X_t(\dexp_{-x_t}\gamma_{t,k}^\wedge)^2+\sum_{j=1}^{d}H_{t,jk}X_t (\partial_{x_j}\dexp_{-x_t})(\gamma_{t,k}^\wedge).
					\end{split}
				\end{equation*}
				Let $ {C}^\wedge_k\triangleq\sum_{j=1}^{d}H_{j,k} (\partial_{x_j}\dexp_{-x_t})(\gamma^\wedge_k) $. By Lemma~\ref{lemma:support_lemma_2},  we rewrite ${C}_k$ ,
				\begin{equation*}
					\begin{split}
					{C_k}&=\sum_{j=1}^{d}H_{j,k} (\partial_{x_j}\dexpm({-x_t}))\gamma_{t,k}\\
					&=\sum_{i=1}^{\infty}\frac{(-1)^{i}}{(i+1)!}\sum^{i-1}_{r=0}\adm^{r}_{x}\adm_{\gamma_{t,k}}\adm^{i-1-r}_{x}\gamma_{t,k},
					\end{split}
				\end{equation*}
				which completes the proof.
				\hfill $\square$

			\bibliographystyle{unsrt}
			\bibliography{ref}
			%
			%
			%

		\end{document}